\documentclass[12pt]{iopart}

\usepackage{amsthm}
\newtheorem{thm}{Theorem}
\newtheorem{cor}[]{Corollary}

\newcommand{\be}{\begin{equation}}
\newcommand{\ee}{\end{equation}}
\newcommand{\ba}{\begin{eqnarray}}
\newcommand{\ea}{\end{eqnarray}}
\eqnobysec

\begin{document}

\title[Killing tensors in pp-wave spacetimes]{Killing tensors in pp-wave spacetimes}

\author{Aidan J Keane$^\sharp$ and Brian O J Tupper$^\natural$}
\address{$^\sharp$ 87 Carlton Place, Glasgow G5 9TD, Scotland, UK.\\
$^\natural$ Department of Mathematics and Statistics, University of
New Brunswick\\Fredericton, New Brunswick, Canada E3B 5A3.}

\eads{\mailto{aidan@countingthoughts.com}, \mailto{bt32@rogers.com}}

\begin{abstract}
The formal solution of the second order Killing tensor equations for the general pp-wave spacetime
is given. The Killing tensor equations are integrated fully for some specific pp-wave spacetimes.
In particular, the complete solution is given for the conformally flat plane wave spacetimes
and we find that irreducible Killing tensors arise for specific classes. The maximum number of independent
irreducible Killing tensors admitted by a conformally flat plane wave spacetime is shown to be six.
It is shown that every pp-wave spacetime that admits an homothety will admit a Killing tensor
of Koutras type and, with the exception of the singular scale-invariant plane wave spacetimes,
this Killing tensor is irreducible.
\end{abstract}


\section{Introduction}
Let $\mathcal{M}$ denote a four-dimensional spacetime manifold with 
Lorentzian metric $g_{ab}$ and metric connection $\Gamma^a_{bc}$. $R_{abcd}$, $R_{ab}$ and $C_{abcd}$
denote the Riemann curvature tensor, Ricci tensor and Weyl tensor respectively.
$\mathcal{L}_X$ denotes the Lie derivative operator with respect to a vector field $X$ on $\mathcal{M}$
and a semicolon denotes the covariant derivative arising from $g_{ab}$ in the usual way.
Symmetrization of index pairs of a tensor field on $\mathcal{M}$ is indicated
by round brackets, i.e., $T_{(ab)} = \case{1}{2} (T_{ab} + T_{ba})$.
For Lie algebras ${\cal A}$ and ${\cal B}$,  the notation ${\cal A} \supset {\cal B}$ means ${\cal B}$
is a subalgebra of ${\cal A}$.

A vector field $X$ on $\mathcal{M}$ which satisfies
\be
(\mathcal{L}_X g)_{ab} = 0 \iff X_{(a;b)} = 0 \label{killingvector}
\ee
is referred to as a {\it Killing vector field} (KV). The concept of a KV can be 
generalized in a variety of ways. A totally symmetric tensor field $K$ of order 
$r$ on $\mathcal{M}$ satisfying
\be
K_{(a_1 \cdots a_r;{a_{r+1}})} = 0 \label{eqn:generalkillingtensor}
\ee
is referred to as a {\it Killing tensor field} (KT). A KT with $r=1$ is a KV.
KTs are of interest principally because of their 
association with polynomial first integrals of the geodesic equation:
If $t$ is the geodesic tangent vector then the quantity
$K_{a_1 \cdots a_r} t^{a_1} \cdots t^{a_r}$ is a first integral of the geodesic motion.
The set of all KVs on $\mathcal{M}$ form a Lie algebra under the bracket operation.
Similarly, the set of all KTs on $\mathcal{M}$ form a {\it graded algebra} 
under the Schouten-Nijenhuis bracket operation \cite{dolan89}.
In this work we shall restrict attention to the 
{\it second order} KTs, i.e., those which satisfy
\be
K_{(ab;c)} = 0, \qquad K_{ab} = K_{ba} .
\label{eqn:killingtensor}
\ee
The metric tensor itself and all symmetrized products of KVs 
are KTs, as are all linear combinations with constant coefficients, i.e.,
\be
K_{ab} = c_0 g_{ab} 
+ \sum^n_{I = 1} \sum^n_{J = I} c_{IJ} X_{I(a} X_{|J|b)}
\label{eqn:reducible}
\ee
is a KT, where $J \ge I$ and $n$ is the dimension of the Lie algebra of KVs. KTs which can be written in
the form (\ref{eqn:reducible}) are known as {\it reducible}, otherwise they are {\it irreducible}.

The maximum number of independent KTs admitted by a four-dimensional spacetime
is 50 and this maximum number is attained if and only if the spacetime is of 
constant curvature \cite{katzin65}, in which case the 50 KTs are reducible \cite{hauser75a}.
The Kerr metric \cite{kerr63} is probably the most well known and interesting example of a
spacetime admitting an irreducible KT \cite{carter68}, \cite{walker70}.
KTs are admitted by other Petrov type $D$ spacetimes \cite{walker70} - \cite{hughston73}.
Kimura \cite{kimura76} - \cite{kimura79} investigated KTs in static spherically
symmetric spacetimes having spatial parts of non-constant curvature: Of particular note are
spacetimes that admit 8 and 11 independent irreducible KTs \cite{kimura79}.
Hauser and Malhiot made a similar study under the assumption that the KTs are independent
of time \cite{hauser74}.
The Taub-NUT spacetime admits four KTs \cite{rietdijk} and the Euclidean Taub-NUT
spacetime admits three KTs \cite{visinescu09}.

It is of interest to investigate the existence of KTs in other spacetimes and in this paper
we investigate KTs, in particular irreducible KTs, in pp-wave spacetimes. There has been
some previous work in this area:
Cosgrove \cite{cosgrove78} considered stationary axisymmetric vacuum spacetimes and has found some
examples of KTs in pp-wave spacetimes.
Our motivation for considering KTs in pp-wave spacetimes is twofold:
(i) A particular class of pp-wave spacetimes, the plane wave spacetimes, exhibit a high
degree of symmetry/conformal symmetry being either of Petrov type $O$ or $N$. The type $O$
spacetimes admit the maximum conformal symmetry and the type $N$ plane wave spacetimes
admit the highest degree of conformal symmetry below the type $O$ \cite{DefriseCarter}, \cite{HallSteele}.
Thus, it is natural to consider this class of spacetime since it is likely to admit
further symmetries.
(ii) Physically, pp-wave spacetimes represent radiation moving at the speed of light,
and so are of particular interest in the theory of gravitational radiation.
Further, gravitational plane wave spacetimes have applications in String Theory \cite{blau06} and arise
naturally as the Penrose Limit \cite{penrose76} of any spacetime, see also \cite{blau06}.

We shall now state precisely what we mean by a pp-wave spacetime.
We define a pp-wave spacetime to be a non-flat spacetime which admits a covariantly constant,
nowhere zero, null bivector. The line element for such a spacetime can be written \cite{ehlers62}
\be
ds^2 =  - 2 du dv - 2 H(u,y,z) du^2 + dy^2 + dz^2.
\label{eqn:ppw}
\ee
A pp-wave spacetime admits a covariantly constant, nowhere zero, null vector field $k$,
which is necessarily a KV, and has the form
\be
k^a = \delta^a_v, \qquad k_a = - \delta^u_a.
\label{eqn:kvk}
\ee
Sippel and Goenner \cite{sippel86} give the form of the Riemann, Weyl and Ricci tensors for the 
spacetime with line element (\ref{eqn:ppw}), the latter being
\[
R_{ab} = (H_{,yy} + H_{,zz}) k_a k_b.
\]
It is clear from this that vacuum and pure radiation fields (with the possibility for null
electromagnetic fields as a special case) can occur. The weak and dominant energy conditions \cite{hawking1973} are
satisfied if
\be
H_{,yy} + H_{,zz} \ge 0.
\label{eq:energycondition}
\ee
The spacetime (\ref{eqn:ppw}) is vacuum if
$H_{,yy} + H_{,zz} = 0$ and conformally flat if $H_{,yy} = H_{,zz}$ and $H_{,yz}=0$. The pp-wave spacetime
is of Petrov type $N$ or $O$, which can be deduced from the form of the Weyl tensor \cite{sippel86}, and
the vacuum case cannot occur for type $O$.
The Riemann curvature tensor satisfies $R_{abcd} k^d = 0$ and if the Weyl tensor is nowhere
zero then $R_{abcd} k^d = C_{abcd} k^d = 0$ and $k$ is a repeated principal null direction of the Weyl tensor.
The similarity of the Weyl tensor for the type $N$ pp-wave spacetime and the electromagnetic
field tensor for electromagnetic plane waves permits the interpretation as {\it gravitational waves};
The vanishing of the expansion and twist of the rays justifies the term {\it plane-fronted}, and
the constancy of $k$ implies {\it parallel rays}: Hence the designation of the fields as
{\it plane-fronted waves with parallel rays}, or {\it pp-waves}.
We note that in \cite{KSMH2} a pp-wave is defined to be a spacetime admitting only the
covariantly constant, nowhere zero, null vector field $k$ and that the imposition of the
conditions on the energy-momentum tensor of the types above are required in order to
obtain the line element of the form (\ref{eqn:ppw}). Ehlers and Kundt \cite{ehlers62}
defined a pp-wave to be a vacuum spacetime. See \cite{KSMH2}, \cite{griffiths09} for an overview.

The wave interpretation permits one to define an amplitude and polarization
for the type $N$ pp-wave spacetimes \cite{ehlers62}. (Type $O$ spacetimes have vanishing Weyl tensor and
hence zero amplitude.) A pp-wave spacetime is said to be a {\it plane wave spacetime} if the amplitude
is constant in every wavefront and in this case the metric function can be written
\be
2H = A(u) y^2 + 2B(u) yz + C(u) z^2
\label{eq:hplanewave}
\ee
where $A$, $B$ and $C$ are arbitrary functions. 
A conformally flat pp-wave spacetime is necessarily a plane wave spacetime and has metric function given by
\be
2H = A(u) (y^2 + z^2)
\label{eq:hcfplanewave}
\ee
where $A$ is an arbitrary function.

We now state some preliminary geometrical results. 
A vector field $X$ is said to be a {\it conformal Killing vector field} (CKV) if and only if
\be
(\mathcal{L}_X g)_{ab} = 2 \phi g_{ab}
\label{eq:ckv}
\ee
where $\phi$ is some function of the coordinates ({\it conformal scalar}). When 
$\phi$ is not constant the CKV is said to be {\it proper}, and if $\phi_{;ab}=0$ the
CKV is a {\it special} CKV (SCKV). When $\phi$ is a constant, $X$ is a {\it homothetic vector field} (HKV) and
when the constant $\phi$ is non-zero $X$ is a {\it proper} HKV.
When $\phi = 0$, $X$ is a KV as mentioned above. The set of all CKV (respectively, SCKV, HKV and KV) form a finite
dimensional Lie algebra denoted by ${\cal C}$ (respectively, ${\cal S}$, ${\cal H}$ and ${\cal G}$).
Koutras \cite{koutras92} devised an algorithm to find KTs using CKVs and this algorithm
was generalized by Rani, Edgar and Barnes \cite{rani03}, \cite{edgar04}:
Given a pair of CKVs $X$, $Y$ satisfying
\[
(\mathcal{L}_X g)_{ab} = 2 \phi g_{ab}, \qquad (\mathcal{L}_Y g)_{ab} = 2 \psi g_{ab}
\]
then if the quantity $\phi Y_a + \psi X_a$ is a gradient, i.e., $\phi Y_a + \psi X_a = \Phi_{,a}$ for some scalar $\Phi$,
it follows that the tensor field
\[
L_{ab} = X_{(a} Y_{b)} - \Phi g_{ab}
\]
is a KT. Special cases are dealt with in the theorems and corollaries given in \cite{rani03}.
This algorithm does not, nor claims to, produce irreducible KTs in general.
We note that, for a KV $X$ and a KT $K$ it is straightforward to show that the tensor 
$(\mathcal{L}_X K)_{ab}$ is a KT. This can be regarded as following from the more
general result given in \cite{dolan89} where it is shown that the graded algebra of KTs has
the same structure constants as the corresponding Poisson bracket Lie algebra of first integrals.
We note that, for a pp-wave spacetime, it follows from (\ref{eqn:killingtensor}) and (\ref{eqn:kvk}) that
\[
\mathcal{L}_k K_{ab} = - \mathcal{L}_{\bar{k}} g_{ab}
\qquad
\mbox{where} \qquad \bar{k}_a = k^b K_{ab}.
\]

For details of the isometry and conformal algebras of the pp-wave spacetimes see \cite{ehlers62}, \cite{sippel86},
\cite{maartens91} and \cite{keanetupper04}. Here we note that the general plane wave spacetime (\ref{eq:hplanewave})
admits an ${\cal H}_6 \supset {\cal G}_5$ and the general conformally flat spacetime (\ref{eq:hcfplanewave}) admits
a ${\cal C}_{15} \supset {\cal H}_7 \supset {\cal G}_6$. In the conformally flat case the underlying conformal
algebra is the $so(4,2)$ conformal algebra of Minkowski spacetime.
Of particular relevance for this work are the plane wave spacetimes which admit additional symmetries, i.e.,
\be
2H = u^{-2} (a y^2 + 2 b yz + c z^2)
\label{eq:planewavesingular}
\ee
and
\be
2H = a y^2 + 2 b yz + c z^2
\label{eq:sgclass13}
\ee
where $a$, $b$ and $c$ are arbitrary constants (not all zero). Note that a transformation of the $y$ and $z$
coordinates allows us to set $b = 0$ in both (\ref{eq:planewavesingular}) and (\ref{eq:sgclass13}).
For arbitrary values of $a$, $b$ and $c$, the metric functions (\ref{eq:planewavesingular}) and
(\ref{eq:sgclass13}) are Sippel and Goenner classes 11 and 13 respectively and both admit a
${\cal H}_7 \supset {\cal G}_6$.
When $a = c$ and $b = 0$, (\ref{eq:planewavesingular}) and
(\ref{eq:sgclass13}) are Sippel and Goenner classes
16 and 17 respectively, are conformally flat and admit a ${\cal C}_{15} \supset {\cal H}_8 \supset {\cal G}_7$.
Plane wave spacetimes with metric function (\ref{eq:planewavesingular}) are {\it singular scale-invariant} plane waves
and those with metric function (\ref{eq:sgclass13}) are {\it symmetric} plane waves \cite{blau06}.

The Koutras algorithm is an indirect method to construct KTs and, as we have pointed out,
does not produce irreducible KTs in general. We are interested in irreducible KTs and particularly
those which cannot be obtained from the Koutras algorithm.
However, we shall now see that if a pp-wave spacetime
admits a homothety then the Koutras KT arises naturally. 

\begin{thm}\label{thm1}
A pp-wave spacetime which admits an HKV $Y$ will admit a KT, which will be irreducible in general.
This KT is obtained from $Y$ and the KV $k$ via the Koutras algorithm.
The only pp-wave spacetime for which this KT is reducible is the plane wave spacetime with metric function
(\ref{eq:planewavesingular}).
\end{thm}

\begin{proof}
Given the pairs $(k, \phi = 0)$, $( Y, \psi = constant)$ then $\phi Y_a + \psi k_a = - \psi \delta_a^u = \Phi_{,a}$,
i.e., $\Phi = -\psi u$ and the Koutras algorithm generates the KT
\be
L_{ab} = k_{(a} Y_{b)} + \psi u g_{ab}.
\label{eq:KoutrasKT-kH}
\ee
Reference \cite{keanetupper04} gives the expressions for the most general HKV and KV in a pp-wave
spacetime: The general HKV is given by $\alpha Z + X$ where $\alpha$ is a constant and
$Z$ and $X$ are given by equations (15) and (16) respectively in \cite{keanetupper04};
the general KV is given by $X$. Using these expressions one can write down the most general
reducible KT formed from a sum of the metric tensor $g_{ab}$ and the symmetrized product of
the general KV $X$, and comparison with (\ref{eq:KoutrasKT-kH}) leads one to the conclusion that
the metric function $H$ must have the form of a plane wave (\ref{eq:hplanewave}). Subsequent application
of the CKV equation (\ref{eq:ckv}) leads to a metric function of the form (\ref{eq:planewavesingular}).

\end{proof}

\begin{cor}\label{cor1}
All plane wave spacetimes admit a KT of Koutras type (\ref{eq:KoutrasKT-kH}) and,
with the exception of those with metric function (\ref{eq:planewavesingular}), this KT is irreducible.
\end{cor}

The general form of the KT components for a general pp-wave spacetime are given in the appendix.
Despite the large number and apparent complexity of the general equations we are able to obtain
explicit solutions for a selection of specific pp-wave spacetimes and these are presented in
section \ref{sec:ppwavespactimes}.
Some examples of plane wave spacetimes are given in section \ref{sec:planewavespacetimes}.
In section \ref{sec:cfplanewavespacetimes} we solve for the KTs explicitly for the conformally
flat plane wave spacetimes. The maximum number of independent irreducible KTs admitted by a conformally flat plane wave spacetime is shown to be six but some admit none at all.

\section{pp-wave spacetimes}
\label{sec:ppwavespactimes}
The full KT equations are given in the appendix. We present two examples.

\subsection*{Example 1}
The type $Biv$ pp-wave spacetime of \cite{keanetupper04} has metric
\[
ds^2 = -2dudv -2 l (\alpha z - \beta y)^{-2} du^2 + dy^2 + dz^2
\]
where $l$, $\alpha$ and $\beta$ are constants such that $\alpha^2 + \beta^2 \ne 0$.
A coordinate transformation puts this metric in the form
\[
ds^2 = -2dudv -2 z^{-2} du^2 + dy^2 + dz^2
\]
which is a 1+3 spacetime \cite{carot2008}. 
This spacetime admits a conformal algebra ${\cal S}_6 \supset {\cal H}_5 \supset {\cal G}_4$ with basis
\ba
X_1 = \partial_v, \qquad X_2 = \partial_u, \qquad X_3 = \partial_y, \qquad X_4 = y \partial_v + u \partial_y
\nonumber\\
X_5 = 2u \partial_u + y \partial_y + z \partial_z,
\qquad
X_6 = u^2 \partial_u + \case{1}{2} (y^2 + z^2) \partial_v + u (y \partial_y + z \partial_z).
\nonumber
\ea
Solving the KT equations (\ref{eq:ppwde2}) - (\ref{eq:ppwde10}) we find that there are 16 independent KTs corresponding to the 16 arbitrary constants of which five are irreducible KTs.
The irreducible KTs are
\ba
(K_1)_{ab} = -2 y^2 z^{-2} \delta_{(a}^u \delta_{b)}^u - z^2 \delta_{(a}^y \delta_{b)}^y
+ 2 yz \delta_{(a}^y \delta_{b)}^z - y^2 \delta_{(a}^z \delta_{b)}^z
\nonumber\\
(K_2)_{ab} = 2yz^{-2} \delta_{(a}^u \delta_{b)}^u - z \delta_{(a}^y \delta_{b)}^z
+ y \delta_{(a}^z \delta_{b)}^z
\nonumber\\
(K_3)_{ab} = 2 u yz^{-2} \delta_{(a}^u \delta_{b)}^u + z^2 \delta_{(a}^u \delta_{b)}^y
-yz \delta_{(a}^u \delta_{b)}^z - uz \delta_{(a}^y \delta_{b)}^z + uy \delta_{(a}^z \delta_{b)}^z
\nonumber\\
(K_4)_{ab} = 2 u z^{-2} \delta_{(a}^u \delta_{b)}^u - z \delta_{(a}^u \delta_{b)}^z
+ u \delta_{(a}^z \delta_{b)}^z
\nonumber\\
(K_5)_{ab} = (z^2 + 2u^2 z^{-2}) \delta_{(a}^u \delta_{b)}^u
- 2uz \delta_{(a}^u \delta_{b)}^z + u^2 \delta_{(a}^z \delta_{b)}^z.
\nonumber
\ea
We note that $K_2$, $K_4$ and $K_5$ can be derived from the Koutras algorithm (in combination 
with reducible KTs) whereas $K_1$ and $K_3$ cannot.

\subsection*{Example 2}
The metric of the isometry class 9 pp-wave spacetime of \cite{sippel86} can be written
\[
ds^2 = -2dudv -2 \exp(y) du^2 + dy^2 + dz^2.
\]
This spacetime admits a ${\cal G}_5$ and no irreducible KTs.

Thus there exists pp-wave spacetimes which admit no irreducible KTs.

\section{Plane wave spacetimes}
\label{sec:planewavespacetimes}
Let us consider the special case of the plane wave spacetime (\ref{eq:hplanewave}).
For an arbitrary $A$, $B$ and $C$ this spacetime admits a ${\cal H}_6 \supset {\cal G}_5$ with basis (see \cite{keanetupper04})
\ba
X_1 = \partial_v,
\qquad
X_2, X_3 = f_{,u} y \partial_v + f \partial_y,
\qquad
X_4, X_5 = g_{,u} z \partial_v + g \partial_z
\nonumber\\
X_6 = 2v \partial_v + y \partial_y + z \partial_z
\label{eq:planewavespacetimeh6basis}
\ea
where the functions $f$, $g$ satisfy
\ba
f_{,uu} + Af + Bg = 0, \qquad g_{,uu} + Bf + Cg = 0.
\label{eq:planewavekvconditions}
\ea
The $X_1, \ldots , X_5$ are KVs and $X_6$ is a proper HKV with $\phi = 1$.
The KT components for the general plane wave spacetimes are obtained from
the equations in the appendix with $H$ given by (\ref{eq:hplanewave}) and, as a consequence,
$\sigma = \zeta = \mu = \epsilon = 0$.
The Koutras KT arising from $X_1$ and $X_6$ is given by
\be
L_{ab} = 2(v - u H) \delta_{(a}^u \delta_{b)}^u
- y \delta_{(a}^u \delta_{b)}^y - z \delta_{(a}^u \delta_{b)}^z
-2u \delta_{(a}^u \delta_{b)}^v + u (\delta_{(a}^y \delta_{b)}^y + \delta_{(a}^z \delta_{b)}^z).
\label{eq:generalplanewavekoutraskt}
\ee
When $A = - C$ the spacetime is vacuum and when $A = C$ and $B=0$ the spacetime is conformally flat,
the latter being dealt with in section \ref{sec:cfplanewavespacetimes}.

\begin{thm}\label{thm:singularplanewave}
The singular scale-invariant plane wave spacetime given by (\ref{eq:planewavesingular}) in general
admits no irreducible KTs.  The only exception occurs in the case of the conformally flat plane wave
spacetime with metric function
\be
2H = \case{3}{16} u^{-2} (y^2 + z^2)
\label{eq:singularplanewave3over16}
\ee
in which case there are six independent irreducible KTs.
\end{thm}

\begin{proof}
A straightforward but lengthy calculation using (\ref{eq:planewavesingular}) in the KT equations
gives the general result.
The plane wave with metric function (\ref{eq:singularplanewave3over16}) arises in the
analysis of the general conformally flat plane wave spacetimes in section
\ref{sec:cfplanewavespacetimes} where the second part of the theorem is proved.
We note that the metric function $2H = - \case{3}{4} u^{-2}(y^2 + z^2)$ will admit an irreducible KT,
however we discard this solution because it does not satisfy the energy conditions (\ref{eq:energycondition}).
\end{proof}

\begin{cor}
\label{cor:singularscaleinvariantvacuumplanewave}
The singular scale-invariant vacuum plane wave spacetimes, i.e., those with metric function
\be
2H = \kappa u^{-2} (y^2 - z^2)
\ee
where $\kappa$ is a constant, admit no irreducible KTs.
\end{cor}

\subsection*{Example 3}
The vacuum plane wave spacetime with $2H = y^2 - z^2$, i.e.,
\be
ds^2 =  - 2 du dv - (y^2 - z^2) du^2 + dy^2 + dz^2
\label{eq:example3}
\ee
admits an ${\cal H}_7 \supset {\cal G}_6$ composed of (\ref{eq:planewavespacetimeh6basis}) and the extra KV
\[
X_7 = \partial_u.
\]
In this case the independent solutions of (\ref{eq:planewavekvconditions}) are
\[
f_1 = \sin u, \qquad f_2 = \cos u, \qquad g_1 = \sinh u, \qquad g_2 = \cosh u.
\]
The solution of the KT equations involves 22 independent arbitrary constants.
However, there are only 21 independent reducible KTs: There are 21 symmetrized products of
the KVs and the metric tensor but the metric tensor is a linear combination of five of the
symmetrized products of KVs, i.e.,
\[
g_{ab} = -2 X_{1(a} X_{7b)} + X_{3(a} X_{3b)} + X_{2(a} X_{2b)}
+ X_{4(a} X_{4b)} - X_{5(a} X_{5b)}.
\]
The irreducible KT is the Koutras KT arising from the KV $X_1 = k = \partial_v$ and the HKV
$X_6$, i.e.,
\[
L_{ab} = [2v - u (y^2 - z^2)] \delta_{(a}^u \delta_{b)}^u
- y \delta_{(a}^u \delta_{b)}^y - z \delta_{(a}^u \delta_{b)}^z
-2u \delta_{(a}^u \delta_{b)}^v + u (\delta_{(a}^y \delta_{b)}^y + \delta_{(a}^z \delta_{b)}^z).
\]

\section{Conformally flat plane wave spacetimes}
\label{sec:cfplanewavespacetimes}
These spacetimes have metric function given by (\ref{eq:hcfplanewave}).
Since we are primarily interested in irreducible KTs and whether
the KTs can be obtained from the Koutras algorithm, we begin by writing down a basis for
the CKV of this spacetime and, where appropriate, the non-zero conformal scalars $\phi$.
The covariant components of the CKVs are given in terms of the functions $f_1$ and $f_2$
which are two independent solutions of
\ba
f_{,uu} + Af = 0.
\label{eq:fdiffuu}
\ea
The components are as follows
\ba
X_{1a} = - \delta^u_a,
\qquad
X_{2a} = z \delta^y_a - y \delta^z_a,
\nonumber\\
X_{3a} = -f_{1,u} y \delta^u_a + f_1 \delta^y_a,
\qquad
X_{4a} = -f_{2,u} y \delta^u_a + f_2 \delta^y_a
\nonumber\\
X_{5a} = -f_{1,u} z \delta^u_a + f_1 \delta^z_a,
\qquad
X_{6a} = -f_{2,u} z \delta^u_a + f_2 \delta^z_a
\nonumber\\
X_{7a} = -2v \delta^u_a + y \delta^y_a + z \delta^z_a, \qquad \phi = 1
\nonumber\\
X_{8a} = [- \case{1}{4} A (y^2 + z^2)^2 - v^2] \delta^u_a - \case{1}{2} (y^2 + z^2) \delta^v_a + v (y \delta^y_a + z \delta^z_a),
\qquad \phi = v
\nonumber\\
\fl X_{9a} = [-\case{1}{2} A f_1 y (y^2 + z^2) - f_{1,u} vy] \delta^u_a - f_1 y \delta^v_a
+ [\case{1}{2} f_{1,u} (y^2 - z^2) + f_1 v]\delta^y_a + f_{1,u} yz\delta^z_a
\nonumber\\
\phi = f_{1,u} y
\nonumber\\
\fl X_{10a} = [-\case{1}{2} A f_2 y (y^2 + z^2) - f_{2,u} vy] \delta^u_a - f_2 y \delta^v_a
+ [\case{1}{2} f_{2,u} (y^2 - z^2) + f_2 v]\delta^y_a + f_{2,u} yz\delta^z_a
\nonumber\\
\phi = f_{2,u} y
\nonumber\\
\fl X_{11a} = [-\case{1}{2} A f_1 z (y^2 + z^2) - f_{1,u} vz] \delta^u_a - f_1 z \delta^v_a
+ f_{1,u} yz\delta^y_a + [\case{1}{2} f_{1,u} (z^2 - y^2) + f_1 v]\delta^z_a
\nonumber\\
\phi = f_{1,u} z
\nonumber\\
\fl X_{12a} = [-\case{1}{2} A f_2 z (y^2 + z^2) - f_{2,u} vz] \delta^u_a - f_2 z \delta^v_a
+ f_{2,u} yz\delta^y_a + [\case{1}{2} f_{2,u} (z^2 - y^2) + f_2 v]\delta^z_a
\nonumber\\
\phi = f_{2,u} z
\nonumber\\
\fl X_{13a} = -\case{1}{2} (A {f_1}^2 + {f_{1,u}}^2) (y^2 + z^2) \delta^u_a - {f_1}^2 \delta^v_a
+ f_1 f_{1,u} (y \delta^y_a + z \delta^z_a),
\qquad \phi = f_1 f_{1,u}
\nonumber\\
\fl X_{14a} = -\case{1}{2} (A {f_2}^2 + {f_{2,u}}^2) (y^2 + z^2) \delta^u_a - {f_2}^2 \delta^v_a
+ f_2 f_{2,u} (y \delta^y_a + z \delta^z_a),
\qquad \phi = f_2 f_{2,u}
\nonumber\\
\fl X_{15a} = -\case{1}{2}(A f_1 f_2 + f_{1,u} f_{2,u}) (y^2 + z^2) \delta^u_a - f_1 f_2 \delta^v_a
+ \case{1}{2} (f_{1,u} f_2 + f_1 f_{2,u}) (y \delta^y_a + z \delta^z_a)
\nonumber\\
\phi = \case{1}{2}(f_1 f_{2,u} + f_2 f_{1,u}).
\nonumber
\ea
Thus, for a conformally flat plane wave spacetime, the conformal symmetries depend only
on the independent solutions of the differential equation (\ref{eq:fdiffuu}).
$X_1, \ldots , X_6$ are KVs, $X_7$ is a proper HKV and, in general, $X_8, \ldots , X_{15}$ are 
proper CKV.

There are some special cases of note as identified by Sippel and Goenner \cite{sippel86}:
When $A$ is constant, $X_{15}$ is a KV, which we shall denote as $Z = \partial_u$ and
\[
Z_a = -A (y^2 + z^2) \delta^u_a - \delta^v_a.
\]
When $A = \kappa u^{-2}$, there is also an extra KV: For $\kappa < 1/4$, $X_{15}$ is an HKV;
For $\kappa = 1/4$, $X_{13}$ is an HKV; For $\kappa > 1/4$, $X_{13} + X_{14}$ is an HKV.
In each case taking a linear combination of the HKV with $X_7$ allows us to replace the HKV
with the KV $Y = u \partial_u - v \partial_v$. The covariant components of $Y$ are
\[
Y_a = [v - \kappa u^{-1} (y^2 + z^2)] \delta^u_a - u \delta^v_a.
\]
The KT components for the conformally flat plane wave spacetimes are obtained from
the equations in the appendix with $H$ given by (\ref{eq:hcfplanewave}) and, as a consequence,
$\sigma = \zeta = \mu = \epsilon = 0$.
In the case of the conformally flat plane wave spacetimes the KT equations separate into independent groups.
Most of the groups lead only to reducible KTs but the following five groups lead to irreducible KTs.

\begin{enumerate}

\item Equations involving $\rho$, $\Psi$, $\Lambda$ only
\ba
\rho + \Psi + \Lambda = 0
\label{eq:cfpw1}\\
\rho_{,uu} = \Psi_{,uu}
\label{eq:cfpw2}\\
\rho_{,uuu} = - A_{,u} (\rho + \Psi) - 2A (\rho_{,u} + \Psi_{,u})
\\
3 \rho_{,uu} + 2A (\rho + \Psi) = 0
\label{eq:cfpw4}\\
8A \rho_{,u} + 3 A_{,u} \rho - A_{,u} \Psi = 0
\label{eq:cfpw5}\\
8A \Psi_{,u} + 3 A_{,u} \Psi - A_{,u} \rho = 0
\label{eq:cfpw6}\\
10A \rho_{,uu} + A_{,u} \rho_{,u} + 4A_{,u} \Psi_{,u} + A_{,uu} \Psi + 4A^2 (\rho + \Psi) = 0
\\
2 A \rho_{,uu} - 7 A_{,u} \rho_{,u} + 2A_{,u} \Psi_{,u}  + A_{,uu} (\Psi - 2 \rho) + 4A^2 (\rho + \Psi) = 0.
\label{eq:cfpw8}
\ea

\item Equations involving $\tau$, $\omega$, $\theta$, $\pi$ only
\ba
\theta_{,uuu} = 0
\label{eq:cfpw9}
\\
3 \pi_{,u} = 2A(\tau + \omega)
\label{eq:cfpw10}
\\
\tau_{,uuu} = -2A_{,u} \tau - 4A\tau_{,u} - A_{,u} \theta_{,u} - 2A \theta_{,uu}
\label{eq:cfpw11}\\
\omega_{,uuu} = -2A_{,u} \omega - 4A \omega_{,u} - A_{,u} \theta_{,u} - 2A \theta_{,uu}
\\
3 \tau_{,uu} + 4A \tau = 2 A_{,u} \theta + 2A \theta_{,u}
\label{eq:cfpw13}\\
3 \omega_{,uu} + 4A \omega = 2 A_{,u} \theta + 2A \theta_{,u}
\label{eq:cfpw14}\\
24A (A \theta_{,u} + A_{,u} \theta) = 16 A \tau_{,uu} + 5 A_{,u} \tau_{,u} + A_{,uu} \tau + 16 A^2 \tau
\label{eq:cfpw15}\\
8A (A \theta_{,u} + A_{,u} \theta) = 2A (\tau + \omega)_{,uu} - \case{1}{2} A_{,u} (\tau + \omega)_{,u}
\nonumber\\
- \case{1}{2} A_{,uu} (\tau + \omega) + \pi_{,uuu} + 4 A \pi_{,u}
\\
4A(\tau - \omega)_{,u} + A_{,u} (\tau - \omega) = 0
\label{eq:cfpw17}\\
24A (A \theta_{,u} + A_{,u} \theta) = 16 A \omega_{,uu} + 5 A_{,u} \omega_{,u} + A_{,uu} \omega + 16 A^2 \omega.
\label{eq:cfpw18}
\ea

\item Equations involving $\xi$, $\Sigma$, $q$ and $s$
\ba
\xi_{,uuu} + 2A_{,u} \xi + 4A \xi_{,u} + 2A_{,u} q + 4 A q_{,u} = 0
\label{eq:cfpw22}\\
\Sigma_{,uuu} + 2A_{,u} \Sigma + 4A \Sigma_{,u} + 2A_{,u} q + 4A q_{,u} = 0
\label{eq:cfpw23}\\
q = -\case{1}{2} \alpha u + \beta
\label{eq:cfpw24}
\\
s = \alpha = \mbox{constant}.
\ea

\item Equations involving $l$, $\Gamma$ and $\Omega$ only
\ba
l_{,uu} + Al = 0
\label{eq:cfpw28}
\\
3 \Omega_{,uu} + A(3 \Omega - 4l) + 2 A_{,u} \Gamma = 0
\label{eq:cfpw29}\\
\Omega_{,uuu} + A_{,u} \Omega + A \Omega_{,u} + 2A \Gamma_{,uu} + 2A_{,u} \Gamma_{,u} + A_{,uu} \Gamma
\nonumber\\
+ 2A^2 \Gamma - A_{,u} l = 0
\label{eq:cfpw30}\\
2l_{,u} + \Gamma_{,uu} + A \Gamma = 0
\label{eq:cfpw31}\\
A_{,u} l + 4 Al_{,u} + 6A \Gamma_{,uu} + 4A_{,u} \Gamma_{,u} + A_{,uu} \Gamma + 6A^2 \Gamma = 0.
\label{eq:cfpw32}
\ea

\item Equations involving $h$, $\nu$ and $\chi$ only. These equations have the same form as those
in the previous group.

\end{enumerate}

We consider each group in turn:

\begin{enumerate}

\item Equations (\ref{eq:cfpw5}) and (\ref{eq:cfpw6}) give
\be
\rho = \rho_0 A^{-1/4} + \Psi_0 A^{-1/2},
\qquad
\Psi = \rho_0 A^{-1/4} - \Psi_0 A^{-1/2}
\ee
where $\rho_0$ and $\Psi_0$ are constants, and equation (\ref{eq:cfpw2}) becomes
\[
\Psi_0 (A^{-1/2})_{,uu} = 0
\]
so that either $\Psi_0 = 0$ and $A$ is an arbitrary non-zero function, or 
$\Psi_0$ is an arbitrary constant and $A = \kappa$ or $\kappa u^{-2}$, where $\kappa$ is an
arbitrary non-zero constant. Equation (\ref{eq:cfpw4})
becomes
\ba
\rho_0 A^{-9/4} & (15 {A_{,u}}^2 - 12 A A_{,uu} + 64 A^3)
\nonumber\\
& + 12 \Psi_0 A^{-5/2} (3 A_{,u}^2 - 2 A A_{,uu}) = 0
\ea
and, on account of the above conditions on $A$ and $\Psi_0$, the second term vanishes and we have either $\rho_0 = 0$ or
\be
12 A A_{,uu} - 15 {A_{,u}}^2 - 64 A^3 = 0.
\label{eq:cfpwcasei1}
\ee
This equation integrates to give
\[
A^{-5/2} {A_{,u}}^2 = \case{64}{3} A^{1/2} + \eta
\]
where $\eta$ is an arbitrary constant and integration of this equation gives
\be
A=\cases{
\case{3}{16} u^{-2} &for $\eta = 0$
\\
(u^2 - \case{4}{3})^{-2} &for $\eta \ne 0$.
\\}
\label{eq:cfpwcasei2}
\ee
There are only two cases satisfying equations (\ref{eq:cfpw1}) - (\ref{eq:cfpw8})
and leading to an irreducible KT:
\begin{enumerate}

\item $A = \case{3}{16} u^{-2}$, $\Psi_0 = 0$ and $\rho_0$ arbitrary.

\item $A = (u^2 - \case{4}{3})^{-2}$, $\Psi_0 = 0$ and $\rho_0$ arbitrary.

\end{enumerate}
In both cases the irreducible KT is given by
\ba
(K_1)_{ab} &= [-\case{4}{3} A^{3/4} yzv + \case{1}{2} A_{,u} A^{-1/4} yz (y^2 + z^2)] \delta_{(a}^u \delta_{b)}^u
\nonumber\\
&+ A^{-5/4} A_{,u} yz \delta_{(a}^u \delta_{b)}^v
\nonumber\\
&+ [- \case{1}{2} A^{-5/4} A_{,u} zv + \case{2}{3} A^{3/4} z (3y^2 + z^2)] \delta_{(a}^u \delta_{b)}^y
\nonumber\\
&+ [- \case{1}{2} A^{-5/4} A_{,u} yv + \case{2}{3} A^{3/4} y (y^2 + 3z^2)] \delta_{(a}^u \delta_{b)}^z
\nonumber\\
&+ 2 A^{-1/4} z \delta_{(a}^v \delta_{b)}^y
+ 2 A^{-1/4} y \delta_{(a}^v \delta_{b)}^z
- 4 A^{-1/4} v \delta_{(a}^y \delta_{b)}^z.
\nonumber
\ea

\item Equation (\ref{eq:cfpw9}) gives
\[
\theta = \theta_1 u^2 + \theta_2 u + \theta_3
\]
where $\theta_1$, $\theta_2$ and $\theta_3$ are constants.
From equation (\ref{eq:cfpw17}) we obtain
\be
\tau - \omega = a_0 A^{-1/4}
\label{eq:cfpwcaseii1}
\ee
where $a_0$ is a constant. Equations (\ref{eq:cfpw13}), (\ref{eq:cfpw14}) and (\ref{eq:cfpwcaseii1}) give
\be
a_0 (12 A A_{,uu} - 15 {A_{,u}}^2 - 64 A^3) = 0.
\label{eq:cfpwcaseii2}
\ee
Using equations (\ref{eq:cfpw13}) and (\ref{eq:cfpw15}), $\theta$ can be eliminated to give
\be
20 A \tau_{,uu} - 5 A_{,u} \tau_{,u} + (32 A^2 - A_{,uu}) \tau = 0
\label{eq:cfpwcaseii3}
\ee
and an identical equation can be derived for $\omega$. Using combinations of equations (\ref{eq:cfpw11}),
(\ref{eq:cfpw13}), (\ref{eq:cfpw15}) and their derivatives we find
\ba
\theta_1 (32 A A_{,u} u^2 + 64 A^2 u + 2A_{,uuu} u^2 + 18 A_{,uu} u + 30 A_{,u})
\nonumber\\
+ \theta_2 (32 A A_{,u} u + 32 A^2 + 2A_{,uuu} u + 9 A_{,uu})
\nonumber\\
+ \theta_3 (32 A A_{,u} + 2 A_{,uuu}) = 0.
\label{eq:cfpwcaseii4}
\ea
These equations can also be used to obtain
\ba
\fl
\case{2}{5} (12 A A_{,uu} - 15 {A_{,u}}^2 - 64 A^3) \tau
=
\nonumber\\
\theta_1 (6 A_{,u} A_{,uu} u^2 - 4A A_{,uuu} u^2 + 42 {A_{,u}}^2 u - 36 A A_{,uu} u - 12 A A_{,u})
\nonumber\\
+ \theta_2 (6 A_{,u} A_{,uu} u - 4A A_{,uuu} u + 21 {A_{,u}}^2 - 18 A A_{,uu})
\nonumber\\
+ \theta_3 (6 A_{,u} A_{,uu} - 4A A_{,uuu})
\label{eq:cfpwcaseii5}
\ea
and an identical equation can be derived for $\omega$.
We note that each of $\theta_1$, $\theta_2$ and $\theta_3$ are either zero or arbitrary constants.
There are two cases to consider: $a_0 \ne 0$ and $a_0 = 0$.
If $a_0 \ne 0$ then (\ref{eq:cfpwcaseii2}) reduces to equation (\ref{eq:cfpwcasei1}) which has
solutions (\ref{eq:cfpwcasei2}).
In the case $A = \case{3}{16} u^{-2}$, equations (\ref{eq:cfpw9}) - (\ref{eq:cfpw18}) and
(\ref{eq:cfpwcaseii1}) give $\theta_3 = 0$ and
\ba
\tau = \tau_1 u^{1/2} - \case{1}{2} \theta_2,
\qquad
&\omega = \omega_1 u^{1/2} - \case{1}{2} \theta_2
\nonumber\\
\theta = \theta_1 u^2 + \theta_2 u,
&\pi = - \case{1}{4} (\tau_1 + \omega_1) u^{-1/2} + \case{1}{8} \theta_2 u^{-1} + \pi_0
\nonumber
\ea
where $\theta_1$, $\theta_2$, $\tau_1$, $\omega_1$ and $\pi_0$ are arbitrary constants and
$\omega_1 = \tau_1 - 2 a_0 / \sqrt[4]{3}$.
The KTs associated with the constants $\pi_0$ and $\theta_1$ are reducible.
However, the KTs associated with the constants $\tau_1$, $\omega_1$ and $\theta_2$ are irreducible and given by,
respectively
\ba
(K_2)_{ab} &= [-\case{1}{8} u^{-3/2} v y^2 - \case{3}{32} u^{-5/2} y^2(y^2 + z^2)] \delta_{(a}^u \delta_{b)}^u
\nonumber\\
&- u^{-1/2} y^2 \delta_{(a}^u \delta_{b)}^v
+ \case{1}{8} u^{-3/2} y^2 z \delta_{(a}^u \delta_{b)}^z
\nonumber\\
&+ (u^{-1/2} vy + \case{1}{4} u^{-3/2} y^3 + \case{1}{8} u^{-3/2} yz^2) \delta_{(a}^u \delta_{b)}^y
\nonumber\\
&+ 2 u^{1/2} y \delta_{(a}^v \delta_{b)}^y
+ (-2u^{1/2} v + \case{1}{4} u^{-1/2} z^2) \delta_{(a}^y \delta_{b)}^y
\nonumber\\
&- \case{1}{2} u^{-1/2} yz \delta_{(a}^y \delta_{b)}^z
+ \case{1}{4} u^{-1/2} y^2 \delta_{(a}^z \delta_{b)}^z
\nonumber
\\
(K_3)_{ab} &= \{ \hbox{as above with} \: y \leftrightarrow z \}
\nonumber
\\
(K_4)_{ab} &= [-\case{3}{16} u^{-2} v (y^2 + z^2) + \case{7}{128} u^{-3} (y^2 + z^2)^2] \delta_{(a}^u \delta_{b)}^u
\nonumber\\
&+ [-2v + \case{3}{4} u^{-1} (y^2 + z^2)] \delta_{(a}^u \delta_{b)}^v
+ 2u \delta_{(a}^v \delta_{b)}^v
\nonumber\\
&- \case{1}{8} u^{-2} y (y^2 + z^2) \delta_{(a}^u \delta_{b)}^y
- \case{1}{8} u^{-2} z (y^2 + z^2) \delta_{(a}^u \delta_{b)}^z
\nonumber\\
&- y \delta_{(a}^v \delta_{b)}^y
- z \delta_{(a}^v \delta_{b)}^z
+ \case{1}{4} u^{-1} yz \delta_{(a}^y \delta_{b)}^z
\nonumber\\
&+ (v - \case{1}{8} u^{-1} z^2) \delta_{(a}^y \delta_{b)}^y
+ (v - \case{1}{8} u^{-1} y^2) \delta_{(a}^z \delta_{b)}^z.
\nonumber
\ea
We will now consider the case $A = (u^2 - \case{4}{3})^{-2}$. 
Inserting the function $A$ into (\ref{eq:cfpwcaseii4}) we find that $\theta_1 = \theta_2 = \theta_3 = 0$ and equations
(\ref{eq:cfpw10}) - (\ref{eq:cfpw18}) give
\ba
\tau = \tau_1 (u^2 - \case{4}{3})^{1/2},
\qquad
\omega = \omega_1 (u^2 - \case{4}{3})^{1/2}
\nonumber\\
\pi = \pi_0 - \case{1}{2} (\tau_1 + \omega_1) u (u^2 - \case{4}{3})^{-1/2}.
\ea
The KT associated with the constant $\pi_0$ is reducible and the KTs associated with
the constants $\tau_1$ and $\omega_1$ are irreducible and given by, respectively
\ba
(K_5)_{ab} &=
\{
-\case{2}{3} (u^2 - \case{4}{3})^2 v
\nonumber\\
&+ \case{4}{3}
[\case{4}{3} y^2 + \case{1}{4} u (u^2 - \case{4}{3}) (y^2 - 3z^2)]
\} (u^2 - \case{4}{3})^{-7/2} y^2 \delta_{(a}^u \delta_{b)}^u
\nonumber\\
&- 2u (u^2 - \case{4}{3})^{-1/2} y^2 \delta_{(a}^u \delta_{b)}^v
+ \case{8}{3} (u^2 - \case{4}{3})^{-3/2} y^2 z \delta_{(a}^u \delta_{b)}^z
\nonumber\\
&+ 2 [u (u^2 - \case{4}{3}) v + \case{1}{3} (2y^2 + z^2)]
(u^2 - \case{4}{3})^{-3/2} y \delta_{(a}^u \delta_{b)}^y
\nonumber\\
&+ 2 (u^2 - \case{4}{3})^{1/2} y \delta_{(a}^v \delta_{b)}^y
\nonumber\\
&+[-2 (u^2 - \case{4}{3})^{1/2} v + \case{1}{2} u(u^2 - \case{4}{3})^{-1/2} z^2] \delta_{(a}^y \delta_{b)}^y
\nonumber\\
&+ \case{1}{2} u (u^2 - \case{4}{3})^{-1/2}
[ y^2 \delta_{(a}^y \delta_{b)}^y - 2 yz \delta_{(a}^y \delta_{b)}^z]
\nonumber
\\
(K_6)_{ab} &= \{ \hbox{as above with} \: y \leftrightarrow z \}.
\nonumber
\ea
We now consider the case $a_0 = 0$, i.e., $\tau = \omega$.
We note that $K_4$ corresponds to this case. For a given $A(u)$,
it is necessary to first solve (\ref{eq:cfpwcaseii3}), then (\ref{eq:cfpwcaseii4}), insert
the $\tau$ and $\theta$ into (\ref{eq:cfpwcaseii5}) and equations (\ref{eq:cfpw9}) - (\ref{eq:cfpw18})
must be checked for consistency. The function $\pi$ can then be solved for using equation (\ref{eq:cfpw10}),
i.e.,
\be
\pi_{,u} = \case{4}{3} A \tau
\ee
and the KT associated with the constant $\pi_0$ arising from this integration is reducible.
We make a number of observations for the case $a_0 = 0$:
Not all of the constants $\theta_1$, $\theta_2$, $\theta_3$ can be arbitrary since equation
(\ref{eq:cfpwcaseii4}) would then require $A = 0$ and it follows that no more than two
irreducible KTs can arise from the set of equations (ii);
There will be no irreducible KTs arising from the set of equations (ii) for an arbitrary
metric function $A$ since (\ref{eq:cfpwcaseii3}), (\ref{eq:cfpwcaseii4}) and (\ref{eq:cfpwcaseii5})
require $\tau = \theta = \pi_{,u} = 0$.

\item In this case the only irreducible KT is the Koutras KT
\[
L_{ab} = X_{1(a} X_{7b)} + u g_{ab}.
\]
As stated in Corollary \ref{cor1}, $L_{ab}$ is reducible when $A = \kappa u^{-2}$.

\item Equation (\ref{eq:cfpw28}) yields
\[
l = l_1 f_1 + l_2 f_2
\]
where $l_1$ and $l_2$ are arbitrary constants. Equation (\ref{eq:cfpw31}) now reads
\[
\Gamma_{,uu} + A \Gamma = - 2l_1 f_{1,u} - 2l_2 f_{2,u} 
\]
and the general solution of this is
\[
\Gamma = \Gamma_1 f_1 + \Gamma_2 f_2 - l_1 u f_1 - l_2 u f_2
\]
where $\Gamma_1$ and $\Gamma_2$ are arbitrary constants. As a result, equation (\ref{eq:cfpw29}) 
gives
\[
\Omega_{,uu} + A \Omega = \case{4}{3} A (l_1 f_1 + l_2 f_2) - \case{2}{3} A_{,u}
(\Gamma_1 f_1 + \Gamma_2 f_2 - l_1 u f_1 - l_2 u f_2)
\]
and the general solution of this equation is
\[
\Omega = \Omega_1 f_1 + \Omega_2 f_2 + \case{2}{3} \Gamma_1 f_{1,u} + \case{2}{3} \Gamma_2 f_{2,u}
- \case{2}{3} l_1 u f_{1,u} - \case{2}{3} l_2 u f_{2,u} 
\]
where $\Omega_1$ and $\Omega_2$ are arbitrary constants. As a result equations (\ref{eq:cfpw31}) and (\ref{eq:cfpw32}) yield
\ba
4 (l_1 f_1 + l_2 f_2)_{,u} (2A + u A_{,u}) + (l_1 f_1 + l_2 f_2) (2A + u A_{,u})_{,u}
\nonumber\\
- 4 (\Gamma_1 f_1 + \Gamma_2 f_2)_{,u} A_{,u} - (\Gamma_1 f_1 + \Gamma_2 f_2) A_{,uu} =0.
\label{eq:cfviiimain equation}
\ea
We have been unable to find the general solution to this equation. However, we note that
if $\Gamma_1 = \Gamma_2 = 0$ the equation yields either $A = \kappa u^{-2}$ where $\kappa$
is an arbitrary constant, resulting in only reducible KTs, or
$l = l_1 f_1 + l_2 f_2 = \alpha |2A + u A_{,u}|^{-1/4}$ where $\alpha$ is an arbitrary constant.
In this case equation (\ref{eq:cfpw28}) gives
\be
5 {X_{,u}}^2 - 4 X X_{,uu} + 16 AX^2 = 0
\qquad
\mbox{where}
\qquad
X = 2A + u A_{,u}.
\label{eq:cfpwviiibconditionA}
\ee
If a function $A(u)$ satisfying this differrential equation can be found then the
corresponding KT is irreducible and is given by
\ba
(K_7)_{ab} &= \{-2 l_{,u} yv + [\case{1}{3} (2A + u A_{,u}) l
- 2 u A l_{,u}] y(y^2 + z^2) \} \delta_{(a}^u \delta_{b)}^u
\nonumber\\
&+ 2 u l_{,u} y \delta_{(a}^u \delta_{b)}^v
+ 2 [l v - \case{1}{3} u A l (3y^2 + 2z^2) - \case{1}{3} l_{,u} z^2] \delta_{(a}^u \delta_{b)}^y
\nonumber\\
&+ \case{2}{3} (l_{,u} - u A l) \delta_{(a}^u \delta_{b)}^z
- 2 u l \delta_{(a}^v \delta_{b)}^y
+ \case{2}{3} u l_{,u} z \delta_{(a}^y \delta_{b)}^z
\nonumber\\
&- \case{2}{3} u l_{,u} y \delta_{(a}^z \delta_{b)}^z.
\nonumber
\ea
If instead, $l_1 = l_2 = 0$, equation (\ref{eq:cfviiimain equation}) yields either $A_{,u} = 0$,
which leads only to reducible KTs, or
$\Gamma = \Gamma_1 f_1 + \Gamma_2 f_2 = \alpha |A_{,u}|^{-1/4}$.
Equation (\ref{eq:cfpw31}) then gives
\be
5 {A_{,uu}}^2 - 4 A_{,u} A_{,uuu} + 16 A A_{,u}^2 = 0.
\label{eq:cfpwviiibconditionA2}
\ee
Any function $A(u)$ satisfying this equation will have a corresponding irreducible KT given by
\ba
(K_8)_{ab} &= - (2A \Gamma_{,u} + \case{1}{3} A_{,u} \Gamma) y (y^2 + z^2) \delta_{(a}^u \delta_{b)}^u
- 2 \Gamma_{,u} y \delta_{(a}^u \delta_{b)}^v
\nonumber\\
&
+ \case{2}{3} A \Gamma (3y^2 + 2 z^2) \delta_{(a}^u \delta_{b)}^y
+ \case{2}{3} A \Gamma yz \delta_{(a}^u \delta_{b)}^z
+ 2 \Gamma \delta_{(a}^v \delta_{b)}^y
\nonumber\\
&- \case{2}{3} \Gamma_{,u} z \delta_{(a}^y \delta_{b)}^z
+ \case{2}{3} \Gamma_{,u} y \delta_{(a}^z \delta_{b)}^z.
\nonumber
\ea
We have found two solutions to (\ref{eq:cfpwviiibconditionA2}), namely:

\begin{enumerate}

\item $A(u) = \case{3}{16} u^{-2}$. This is the only function of the form
$A(u) = \kappa u^{-2}$ that satisfies equation (\ref{eq:cfpwviiibconditionA2}).
The corresponding irreducible KT is $K_8$ with $f_1 = u^{3/4}$, $f_2 = u^{1/4}$,
$\Gamma_2 = 0$. 

\item $A(u) = -[\case{1}{4} + W(-e^{-u})][1 + W(-e^{-u})]^{-4}$ where $W(x)$ is the
Lambert W-function \cite{corless96}. The corresponding irreducible KT is $K_8$
with $f_1$, $f_2$ given by
\[
{f_1}^2 = -(1 + W) W^{-1}, \qquad {f_2}^2 = -(1 + W) W.
\]
Now $A(u) > 0$ for the energy conditions to hold and ${f_1}^2$, ${f_2}^2$ are each
positive so it follows that the solution is valid on the principle branch of the
W-function in the interval $-1 < W(-e^{-u}) < -1/4$ so that $u$ is confined to the 
interval $1 < u < 1.615$, approximately.

\end{enumerate}

\item By symmetry with the case above, the functions $h$, $\nu$, $\chi$ lead to the
same solutions $A(u)$ but with new irreducible KTs, $K_{9}$ and $K_{10}$, that
can be derived from $K_7$ and $K_8$ by interchanging the coordinates $y$ and $z$.

\end{enumerate}

The results of this section can be summarized in the following theorem.

\begin{thm}\label{thm2}
An arbitrary conformally flat plane wave spacetime will admit only one irreducible KT,
and that KT is the Koutras KT $L_{ab}$. Special subcases exist:

\begin{enumerate}

\item A(u) = $\kappa u^{-2}$, $\kappa = constant$, $\kappa \ne 3/16$. In this case the
Koutras KT $L_{ab}$ is reducible and the spacetime admits no irreducible KTs.

\item A(u) = $(3/16) u^{-2}$. In this case the Koutras KT $L_{ab}$ is reducible and there exist
six independent irreducible KTs, $K_1$, $K_2$, $K_3$, $K_4$, $K_8$ and $K_{10}$.

\item A(u) = $(u^2 - \case{4}{3})^{-2}$. In this case there exist four independent irreducible KTs,
the Koutras KT $L_{ab}$, and $K_1$, $K_5$ and $K_6$.

\item A(u) is not given by any of the above and satisfies (\ref{eq:cfpwcaseii4}) and there 
can be at most three independent KTs, one of which is the irreducible Koutras KT $L_{ab}$.
The other KTs may or may not be irreducible.

\item A(u) is not given by any of the above and satisfies (\ref{eq:cfpwviiibconditionA}) and
there are three irreducible KTs, the Koutras KT $L_{ab}$, $K_7$ and $K_9$.

\item A(u) is not given by any of the above and satisfies (\ref{eq:cfpwviiibconditionA2}) and
there are three irreducible KTs, the Koutras KT $L_{ab}$, $K_8$ and $K_{10}$.

\end{enumerate}
\end{thm}

\begin{cor}
The maximum number of independent irreducible KTs in a conformally flat plane wave spacetime is six.
\end{cor}

We remark that the above results are similar to those found by Kimura \cite{kimura76} - \cite{kimura79}
in that only a very few specific metrics admit irreducible KTs and some of those which do, admit many.
It is worth noting that even the class with $A = constant$ admits only one irreducible KT, i.e., the
Koutras KT $L_{ab}$. Further, we have been unable to find the general solutions corresponding to
the cases (iv), (v), and (vi) of Theorem \ref{thm2}.

The conformally flat plane-wave spacetimes correspond to the Sippel and Goenner spacetimes
of class 15, 16 and 17, which are {\it homogeneous pure radiation} solutions and can also be 
interpreted as Einstein-Maxwell solutions, see \cite{sippel86} for details.
The Sippel and Goenner \cite{sippel86} class 16 spacetime, i.e., $A = constant$, can also be interpreted as an
Einstein-Klein-Gordon solution, see \cite{maartens91} for details.

\section{Discussion}
\label{sec:discussion}
We have given the formal solution of the second order KT equations for the general pp-wave spacetime,
and have presented some noteworthy examples. We note that all physically meaningful
pp-wave spacetimes are subject to the energy condition (\ref{eq:energycondition}). The complete solution is
given for the conformally flat plane wave spacetimes and we have found that irreducible KTs arise
for specific classes. A number of theorems are given regarding the number of irreducible KTs
admitted by pp-wave spacetimes. 
It is worth noting that the technique used in the proof of Theorem \ref{thm1} can be applied
to {\it any} gradient CKV $Y$ and the condition for reducibility determined. Further, one could in principle apply
the reducibility condition to the {\it general} KT components given in the appendix.
So far we have been unable to explain geometrically {\it why} the value $\kappa = 3/16$ is
singled out amongst the singular scale-invariant plane wave spacetimes. To this end we
investigated the geodesic equations: As we have stated already, KTs are of interest principally because of
their association with quadratic first integrals of the geodesic equation. However, this did not provide any
further illumination: The geodesic equations can be integrated without recourse to the use
of the irreducible KTs which is in contrast to the situation in \cite{carter68} where complete
integration of the geodesic equations required the first integral arising from the irreducible KT.
The equation of geodesic deviation is important in the analysis of gravitational
waves and may provide some insight into the problem but we only touch upon geodesic deviation
briefly in what follows.
The general results for KTs in {\it vacuum} pp-wave spacetimes will be presented elsewhere.

The existence of KTs is an interesting topic in its own right, however, we shall now mention some
applications of the results obtained in this paper.

\subsection*{Equation of Geodesic Deviation}
Consider a family of geodesics $x^a = x^a(\lambda)$ in an arbitrary spacetime, 
where $\lambda$ is an affine parameter. Let $t^a = dx^a / d\lambda$ be the tangent vector to a geodesic.
The equation of geodesic deviation is
\[
t^b t^c \nabla_b \nabla_c \xi^a = {R^a}_{bcd} t^b t^c \xi^d
\]
where $\xi^a$ is the vector field representing the geodesic separation.
Let $X_A$, $A = 1, \dots , r$ be a basis for the isometry algebra ${\cal G}_r$.
Further, let $K_{a a_1 \dots a_p}$ be a KT of order $p+1$ and define the field
$W_a = K_{a a_1 \dots a_p} t^{a_1} \dots t^{a_p}$. We shall denote the set of such
fields as $W_B$, $B = 1, \dots , s$. The general solution of the equation of geodesic deviation involves
eight independent solutions and the set consisting of $t$, $\lambda t$, $X_A$ and $W_B$ are solutions
of the equation of geodesic deviation \cite{caviglia82}. For example, in the case of the general plane wave
spacetime (\ref{eq:hplanewave}) we have $r = 5$ and, on account of the existence of the Koutras KT
(\ref{eq:generalplanewavekoutraskt}), $s = 1$. These, together with $t$ and $\lambda t$ provide eight
solutions, however, their independence would have to be verified on a case-by-case basis.

\subsection*{Penrose Limits}
We will now investigate the existence of KTs in the Penrose Limits of two important spacetimes.
These Penrose Limits are derived in \cite{blau06}. First consider the Schwarzschild spacetime
\[
ds^2 = - f(r) dt^2 + f(r)^{-1} dr^2 + r^2 (d \theta^2 + \sin^2 \theta d \phi^2), \qquad f(r) = 1 - 2M/r.
\]
For radial null geodesics the Penrose Limit is Minkowski spacetime. However, the Limit for
the null geodesics at constant $r$ (i.e., the unstable circular orbits at $r = 3M$) is the
type $N$ vacuum plane wave spacetime (\ref{eq:example3}) which admits one irreducible KT.
For the non-radial, non-circular null geodesics the Penrose Limit of the singularity
is given by the type $N$ vacuum plane wave spacetime
\[
ds^2 =  - 2 du dv - \kappa (y^2 - z^2) u^{-2} du^2 + dy^2 + dz^2, \qquad \kappa = 6/25
\]
which, from Corollary \ref{cor:singularscaleinvariantvacuumplanewave}, admits no irreducible KTs.
Now consider the FRW spacetime
\[
ds^2 = -dt^2 + a(t)^2 [dr^2 + f_\epsilon(r)^2(d \theta^2 + \sin^2 \theta d \phi^2)]
\]
where $f_\epsilon(r) = r, \sin r, \sinh r$ for $\epsilon = 0, +1, -1$ respectively, with equation of state
$p(t) = \omega \rho(t)$ where $\omega \ge -1$ is constant. In the case $\omega = -1$ the Penrose Limit is
Minkowski spacetime. In what follows, the constant $h$ is defined by
$h = 2/[3(1 + \omega)]$. For $\omega > -1$ the Penrose Limit for the singularity of the FRW spacetime is given
by the type $O$ plane wave spacetime
\[
ds^2 =  - 2 du dv - \kappa (y^2 + z^2) u^{-2} du^2 + dy^2 + dz^2.
\]
For the case $\epsilon = 0$, and the case $\epsilon = \pm 1$ with $0 < h < 1$, the constant $\kappa = h (1 + h)^{-2}$
and it follows that $0 < \kappa < 1/4$.
Theorem \ref{thm2} states that in general this type $O$ plane wave spacetime will admit no irreducible KTs,
with the exception of the case $\kappa = 3/16$ which admits the maximum number of six. For $\epsilon = 0$ this value
of $\kappa$ corresponds to $h = 3$. 
The value $\kappa = 3/16$ can also arise in the regime $\epsilon = \pm 1$, $h \ge 1$, see \cite{blau06} for details.

\section*{Acknowledgments}
We would like to thank David Lonie and Graham Hall for comments on the first draft of this paper
and we would like to thank the referees for their comments and suggestions.

\appendix
\section*{Appendix}
\setcounter{section}{1}

The non-zero connection coefficients for (\ref{eqn:ppw}) are
\[
\Gamma^v_{uu} = H_{,u}, \qquad
\Gamma^y_{uu} = \Gamma^v_{uy} = \Gamma^v_{yu} = H_{,y}, \qquad
\Gamma^z_{uu} = \Gamma^v_{uz} = \Gamma^v_{zu} = H_{,z}.
\]
For the pp-wave spacetime, (\ref{eqn:killingtensor}) gives
20 independent differential equations. 
11 of these equations are independent of the metric function $H$ and its
derivatives, and 9 equations are dependent upon the metric function $H$.
Direct integration of these equations gives the following expressions for 
the components $K_{ab}$
\ba
K_{uu} = & [\mu H_{,y}   + \epsilon H_{,z} 
- \mu_{,uu} y - \epsilon_{,uu} z + \theta_{,uu} ]v^2 + M v + N
\nonumber\\
K_{uv} = & (\mu_{,u} y + \epsilon_{,u} z - \theta_{,u}) v + D
\nonumber\\
K_{uy} = & -\mu_{,u} v^2 + P v + R, \qquad K_{uz} = - \epsilon_{,u} v^2 + Sv + Q
\nonumber\\
K_{vv} = & W_1, \qquad K_{vy} = \mu v + W_2, \qquad K_{vz} = \epsilon v + W_3
\nonumber\\
K_{yy} = & - 2(\sigma z + \tau)v - \pi z^2 + \chi z + \xi
\nonumber\\
K_{yz} = & (\sigma y + \zeta z + \Lambda)v + \pi yz 
- \case{1}{2} (\chi y + \Omega z) + \Phi
\nonumber\\
K_{zz} = & - 2(\zeta y + \omega)v - \pi y^2 + \Omega y + \Sigma
\nonumber
\ea
where $M$, $N$, $D$, $P$, $Q$, $R$, $S$, $W_1$, $W_2$ and $W_3$ are functions of $u$, $y$ and $z$; 
$\mu$, $\epsilon$, $\theta$, $\sigma$, $\tau$, $\zeta$, $\Psi$, $\Gamma$, $\omega$, 
$\rho$, $\nu$, $\pi$, $\chi$, $\xi$, $\Lambda$, $\Omega$, $\Phi$ and $\Sigma$ are 
functions of $u$ only. The functions $W_1$, $W_2$ and $W_3$ are given by
\ba
W_1 = - 2 \mu y - 2 \epsilon z + 2 \theta, \qquad
W_2 = \sigma yz + \tau y - \zeta z^2 + \Psi z + \Gamma
\nonumber\\
W_3 = \zeta yz + \omega z - \sigma y^2 + \rho y + \nu.
\nonumber
\ea
The functions and constants are governed by the following differential equations:
\ba
M_{,u} = 2H_{,u} (\mu_{,u} y + \epsilon_{,u} z - \theta_{,u})
+ 2 H_{,y} P + 2 H_{,z} S
\label{eq:ppwde2}\\
N_{,u} = 2H_{,u} D + 2H_{,y} R + 2H_{,z} Q
\label{eq:ppwde3}\\
M + 2D_{,u} - 2H_{,u} W_1 - 2H_{,y} W_2 - 2H_{,z} W_3 = 0
\label{eq:ppwde4}\\
P_{,y} = \sigma_{,u} z + \tau_{,u} + 2 \mu H_{,y}
\label{eq:ppwde11}\\
\chi_{,u} z -\pi_{,u} z^2 + \xi_{,u} + 2 R_{,y} - 4H_{,y} W_2 = 0
\label{eq:ppwde12}\\
S_{,z} = \zeta_{,u} y + \omega_{,u} + 2 \epsilon H_{,z} 
\label{eq:ppwde13}\\
\Omega_{,u} y -\pi_{,u} y^2 + \Sigma_{,u} + 2 Q_{,z} - 4H_{,z} W_3 = 0
\label{eq:ppwde14}\\
\Psi_{,u} z + D_{,y} + P + \sigma_{,u} yz + \tau_{,u} y - \zeta_{,u} z^2 + \Gamma_{,u}
- 2 H_{,y} W_1 = 0
\label{eq:ppwde15}\\
\rho_{,u} y + D_{,z} + S + \zeta_{,u} yz + \omega_{,u} z - \sigma_{,u} y^2 + \nu_{,u}
- 2 H_{,z} W_1 = 0
\label{eq:ppwde16}\\
Q_{,y} + R_{,z} + \pi_{,u} yz - \case{1}{2} \chi_{,u} y - \case{1}{2} \Omega_{,u} z + \Phi_{,u}
- 2H_{,y} W_3 - 2H_{,z} W_2 = 0
\label{eq:ppwde2ndlast}\\
\rho + \Lambda + \Psi = 0
\label{eq:ppwdelast}\\
\mu H_{,uy}  + 3 \mu_{,u} H_{,y}
+ \epsilon H_{,uz} + 3 \epsilon_{,u} H_{,z}
- \mu_{,uuu} y - \epsilon_{,uuu} z + \theta_{,uuu} = 0
\label{eq:ppwde1}\\
-3 \mu_{,uu} + \mu H_{,yy} + \epsilon H_{,yz} = 0
\label{eq:ppwde5}\\
-3 \epsilon_{,uu} + \epsilon H_{,zz} + \mu H_{,yz} = 0
\label{eq:ppwde8}\\
S_{,y} + P_{,z} + \sigma_{,u} y + \zeta_{,u} z + \Lambda_{,u}
- 2 \epsilon H_{,y} - 2 \mu H_{,z}  = 0.
\label{eq:ppwde17}\\
\fl M_{,y} + 2 P_{,u} - 4H_{,y} (\mu_{,u} y + \epsilon_{,u} z - \theta_{,u})
-2 \mu H_{,u} 
\nonumber\\
+ 4 H_{,y} (\sigma z + \tau) -2H_{,z} (\sigma y + \zeta z + \Lambda) = 0
\label{eq:ppwde6}\\
\fl N_{,y} + 2 R_{,u} - 4 H_{,y} D - 2H_{,u} W_2
\nonumber\\
- 2H_{,y} (-\pi z^2 + \chi z + \xi)
- 2H_{,z} (\pi yz - \case{1}{2} \chi y - \case{1}{2} \Omega z + \Phi) = 0
\label{eq:ppwde7}\\
\fl M_{,z} + 2 S_{,u} - 4H_{,z} (\mu_{,u} y + \epsilon_{,u} z - \theta_{,u})
-2 \epsilon H_{,u} 
\nonumber\\
- 2 H_{,y} (\sigma y + \zeta z + \Lambda) + 4H_{,z} (\zeta y + \omega) = 0
\label{eq:ppwde9}\\
\fl N_{,z} + 2 Q_{,u} - 4 H_{,z} D - 2H_{,u} W_3
\nonumber\\
- 2H_{,z} (-\pi y^2 + \Omega y + \Sigma)
- 2H_{,y} (\pi yz - \case{1}{2} \chi y - \case{1}{2} \Omega z + \Phi) = 0
\label{eq:ppwde10}
\ea


\section*{References}
\begin{thebibliography}{99}

\bibitem{dolan89} Dolan P, Kladouchou W and Card C 1989 {\it Gen. Rel. Grav.} {\bf 21} 427

\bibitem{katzin65} Katzin GH and Levine J 1965 {\it Tensor, N.S.} {\bf 16} 97

\bibitem{hauser75a} Hauser I and Malhiot RJ 1975 {\it \JMP} {\bf 16} 150
\nonum Hauser I and Malhiot RJ 1975 {\it \JMP} {\bf 16} 1625

\bibitem{kerr63} Kerr R 1963 {\it Phys. Rev. Lett.} {\bf 11} 237

\bibitem{carter68} Carter B 1968 {\it Phys. Rev.} {\bf 174} 1559

\bibitem{walker70} Walker M and Penrose R 1970 {\it Commun. Math. Phys.} {\bf 18} 265

\bibitem{hughston72} Hughston LP, Penrose R, Sommers P and Walker M 1972 {\it Commun. Math. Phys.} {\bf 27} 303

\bibitem{hughston73} Hughston LP and Sommers P 1973 {\it Commun. Math. Phys.} {\bf 32} 147

\bibitem{kimura76} Kimura M 1976 {\it Tensor, N.S.} {\bf 30} 27

\bibitem{kimura77} Kimura M 1977 {\it Tensor, N.S.} {\bf 31} 187

\bibitem{kimura79} Kimura M 1979 {\it Tensor, N.S.} {\bf 33} 123

\bibitem{hauser74} Hauser I and Malhiot RJ 1974 {\it \JMP} {\bf 15} 816

\bibitem{rietdijk} Rietdijk RH and van Holten JW 1996 {\it Nuclear Physics B} {\bf 472} 427

\bibitem{visinescu09} Visinescu M 2009 {\it Journal of Physics: Conference Series} {\bf 189} 012044

\bibitem{cosgrove78} Cosgrove CM 1978 {\it J. Phys. A.} {\bf 11} 2405

\bibitem{DefriseCarter} Defrise-Carter L 1975 {\it Commun. Math. Phys.} {\bf 40} 273

\bibitem{HallSteele} Hall GS and Steele JD 1991 {\it \JMP} {\bf  32} 1847

\bibitem{blau06} Blau M, Borunda M, O'Loughlin M and Papadopoulos G 2004 {\it Class. Quantum Grav.} {\bf 21} L43
\nonum Blau M 2006 {\it Plane Waves and Penrose Limits} (Lecture notes \newline
http://www.blau.itp.unibe.ch/lecturesPP.pdf)

\bibitem{penrose76} Penrose R 1976 {\it Any space-time has a plane wave as a limit}
(Differential geometry and relativity, Reidel, Dordrecht) p271

\bibitem{ehlers62} Ehlers J and Kundt W 1962 {\it Gravitation: An Introduction to Current 
Research}, Edited by L. Witten, (Wiley, New York)
\nonum Jordan P, Ehlers J and Kundt W 1960 {\it Akad. Wiss. Lit. (Mainz) Abhandl. Math. - Nat. Kl.} {\bf 2} 21

\bibitem{sippel86} Sippel R and Goenner H 1986 {\it Gen. Rel. Grav.} {\bf 18} 1229

\bibitem{hawking1973} Hawking SW and Ellis GFR 1973 
{\it The Large Scale Structure Of Space-Time}, (Cambridge: Cambridge University Press)

\bibitem{KSMH2} Stephani H, Kramer D, MacCallum MAH, Hoenselaers C and Herlt E, 2004
{\it Exact Solutions to Einstein's Field Equations} 2nd edn, (Cambridge: Cambridge University Press)

\bibitem{griffiths09} Griffiths JB and Podolsky J 2009
{\it Exact Space-Times in Einstein's General Relativity} (Cambridge: Cambridge University Press)

\bibitem{koutras92} Koutras A 1992 {\it Class. Quantum Grav.} {\bf 9} 1573

\bibitem{rani03} Rani R, Edgar SB and Barnes A 2003 {\it Class. Quantum Grav.} {\bf 20} 1929

\bibitem{edgar04} Edgar SB, Rani R and Barnes A 2004 {\it Proc. Inst. Mathematics of NAS of Ukraine} {\bf 50} part 2, 708

\bibitem{maartens91} Maartens R and Maharaj SD 1991 {\it Class. Quantum Grav.} {\bf 8} 503

\bibitem{keanetupper04} Keane AJ and Tupper BOJ 2004 {\it Class. Quantum Grav.} {\bf 21} 2037

\bibitem{carot2008} Carot J, Keane AJ and Tupper BOJ 2008 {\it Class. Quantum Grav.} {\bf 25} 055002

\bibitem{corless96} Corless RM, Gonnet GH, Hare DEG, Jeffrey DJ and Knuth DE 1996
{\it Advances in Computational Mathematics} {\bf 5} 329

\bibitem{caviglia82} Caviglia G, Zordan C and Salmistraro F 1982 {\it International Journal of Theoretical Physics} {\bf 21} 391

\end{thebibliography}
\end{document}